\documentclass[letterpaper,11pt]{article}

\usepackage{microtype}
\usepackage{hyperref}

\usepackage{cite,mathtools,amssymb,amsthm, latexsym, color,graphicx, xspace, url}

\usepackage{euscript, comment}
\usepackage{enumerate}

\usepackage{xspace}

\usepackage{fullpage}

\newtheorem{theorem}{Theorem}[section]
\newtheorem{lemma}[theorem]{Lemma}
\newtheorem{claim}[theorem]{Claim}

\newtheorem{corollary}[theorem]{Corollary}

\theoremstyle{remark}
\newtheorem{definition}[theorem]{Definition}

\newtheorem*{remark}{Remark}
\newtheorem*{remarks}{Remarks}

\let\eps\varepsilon
\let\phi\varphi
\DeclareMathOperator{\polylog}{polylog}
\DeclarePairedDelimiter\abs{\lvert}{\rvert}
\DeclareMathOperator{\poly}{poly}
\def\A{\EuScript{A}}

\def\F{\EuScript{F}}

\let\Hsave\H
\def\H{\EuScript{H}}

\def\L{\EuScript{L}}

\def\V{\EuScript{V}}

\def\VD{\A^{\shortparallel}}

\def\reals{{\mathbb R}}

\DeclareMathOperator{\EE}{\mathbb{E}}

\newdimen\instindent
\def\institute#1{\gdef\@institute{#1}}

\begin{document}
\begin{titlepage}
  \title{Constructive Polynomial Partitioning for Algebraic Curves in~$\reals^3$ with Applications\thanks{%
      Work on this paper by Boris Aronov has been supported
      by NSF Grants CCF-11-17336, CCF-12-18791, and CCF-15-40656,
      and by BSF grant 2014/170.
      Work on this paper by Esther Ezra has been supported by NSF CAREER under grant CCF:AF 1553354
      and by Grant 824/17 from the Israel Science Foundation.
      Work on this paper by Joshua Zahl was supported by a NSERC Discovery grant.
      A preliminary version of this work was presented in \emph{Proc. 30th Annual ACM-SIAM Sympos. Discrete Algorithms}, 2019~\cite{AEZ-19}.
    }
  }
  
\author{
  Boris Aronov%
  \footnote{Department of Computer Science and Engineering,
    Tandon School of Engineering, New York University, Brooklyn, NY~11201, USA;
    \textsl{boris.aronov@nyu.edu}.
  }
  \and
  Esther Ezra%
  \footnote{Department of Computer Science, Bar-Ilan University, Ramat Gan, Israel;
    \textsl{ezraest@cs.biu.ac.il}.
  }
  \and
  Joshua Zahl\footnote{Department of Mathematics, University of British Columbia, Vancouver, BC~V6T 1Z2, Canada;
    \textsl{jzahl@math.ubc.ca}.}
}



\maketitle

\begin{abstract}
  In 2015, Guth proved that for any set of $k$-dimensional bounded complexity varieties in $\reals^d$ and for any positive integer $D$, there exists a polynomial of degree at most $D$ whose zero set divides $\reals^d$ into open connected sets, so that only a small fraction of the given varieties intersect each of these sets. Guth's result generalized an earlier result of Guth and Katz for points.

  Guth's proof relies on a variant of the Borsuk-Ulam theorem, and for $k>0$, it is unknown how to obtain an explicit
  representation of such a partitioning polynomial and how to construct it efficiently.
  In particular, it is unknown how to effectively construct such a polynomial for bounded-degree algebraic curves
  (or even lines) in ${\reals}^3$.
  
  We present an efficient algorithmic construction for this setting.
  Given a set of $n$ input algebraic curves and a positive integer $D$, we efficiently construct a decomposition of space into
  $O(D^3\log^3{D})$ open ``cells,'' each of which meets $O(n/D^2)$ curves from the input. The construction time is $O(n^2)$.
  For the case of lines in $3$-space we present an improved implementation, whose running time is $O(n^{4/3} \polylog{n})$.
  The constant of proportionality in both time bounds depends on $D$ and the maximum degree of the polynomials
  defining the input curves.

  As an application, we revisit the problem of eliminating depth cycles among non-vertical lines in $3$-space, recently studied by Aronov and Sharir (2018), and show an algorithm that cuts $n$ such lines into $O(n^{3/2+\eps})$ pieces that are depth-cycle free, for any $\eps > 0$.
  The algorithm runs in $O(n^{3/2+\eps})$ time, which is a considerable improvement over the previously known algorithms.
  
\end{abstract}

\thispagestyle{empty}
\end{titlepage}

\section{Introduction}
\label{sec:intro}

\paragraph*{Partitioning Polynomials.}
In \cite{Guth-15}, Guth developed an efficient space decomposition adapted to a set of varieties in Euclidean space. Specifically, he proved the following:

\begin{theorem}[Polynomial Partitioning for Varieties \cite{Guth-15}]
  \label{thm:guth}
  Let $\Gamma$ be a set of $k$-dimensional varieties in ${\reals}^d$, each defined by at most $m$
  polynomials of degree at most $b$. For each $D \ge 1$, there is a~$d$-variate non-zero ``partitioning polynomial'' $f$ of degree
  at most~$D$, so that ${\reals}^d \setminus Z(f)$ is a union of $O(D^d)$ connected components, each of which intersects at most
  $C \frac{\abs{\Gamma}}{D^{d-k}}$ varieties from $\Gamma$. Here $C > 0$ is a constant that depends on~$b$, $m$, and~$d$.
\end{theorem}

In particular, when $\Gamma$ is a set $\L$ of algebraic curves in $\reals^3$ (defined by polynomials of degree at most~$b$),
or just lines, Theorem~\ref{thm:guth} guarantees the existence of a polynomial $f$ of degree at most~$D$ partitioning~$\reals^3$
into $O(D^3)$ connected components, each of which intersect $O(\abs{\L}/D^2)$ curves of~$\L$. Here and throughout this paper, we will think of $b$ as being fixed as $\abs{\L}$ grows, and all implicit constants are allowed to depend on $b$.

Aronov, Miller, and Sharir \cite{AMS-17} used Theorem \ref{thm:guth} to prove that $n$ pairwise disjoint non-vertical triangles in~$\reals^3$ can be cut into $O(n^{3/2 + \eps})$ pieces that form a ``depth order,'' for any $\eps > 0$ (see below for the definition of depth order and a further discussion). This work extended an earlier result of Aronov and Sharir \cite{AS-18}, who proved a bound of $O(n^{3/2} \polylog{n})$, for the analogous problem for pairwise disjoint non-vertical lines in $\reals^3$.
Apart from the $\eps$ loss in the exponent of the triangle bound and $\polylog{n}$ factor in the line bound, these results are  worst-case optimal.

Theorem~\ref{thm:guth} uses a variant of the Borsuk-Ulam theorem to obtain the partitioning polynomial. However, there is no known effective method to construct such a polynomial. Therefore, despite the recent progress on eliminating depth cycles, there is no matching algorithmic bound for the result established in \cite{AMS-17}.
The best known result in this direction is the algorithm presented by De~Berg \cite{deBerg-17}, which exploits a different technique and achieves a suboptimal bound on the number of pieces, where these pieces are triangular fragments (in contrast to the procedure in~\cite{AMS-17}, which yields pieces bounded by algebraic curves). For the case of lines, the work in \cite{AS-18} describes several slow polynomial-time algorithms to compute a depth order,
among which is an approximation algorithm by Aronov, De~Berg, Gray, and Mumford \cite[Theorem 3.1]{AdBGM-08} that produces a set of cuts whose size is larger than that of the smallest possible by only a polylogarithmic factor.

If $\Gamma$ is replaced by a set of curves~$\L$ in~${\reals}^d$ in Theorem~\ref{thm:guth}, then it is easy to find a degree $D$ polynomial $f$ so that \emph{on average} each connected component of ${\reals}^3 \setminus Z(f)$ meets $O(\abs{\L}/D^2)$ curves from $\L$ (indeed, if $f$ is a degree $D$ polynomial, then by Warren \cite[Theorem~2]{Warren-68}, the number of
connected components of~${\reals}^d \setminus Z(f)$ is at most $O(D^3)$; 
every algebraic curve not contained in $Z(f)$ intersects $Z(f)$ in $O(D)$ points,
and thus intersects only $O(D)$ connected components of ${\reals}^d \setminus Z(f)$). Enforcing the property that \emph{every} connected component of ${\reals}^d \setminus Z(f)$ intersects $O(\abs{\L}/D^2)$ curves, however, is much more difficult. In fact, even achieving a more modest
bound, say, of the form~$O(\abs{\L}/D^{1+\eps})$, for some $\eps > 0$, is already challenging.
Using \emph{$\eps$-cuttings}, one can produce a space decomposition, such that each connected component meets roughly
$O\left(\frac{\abs{\L}}{D}\right)$ curves of $\L$ (this bound is larger by an order of magnitude than our target bound),
and such that the total number of curve-connected component intersections is close to $O(\abs{\L} D)$, see, e.g., \cite{ES-05, KS-05}.
However, we are not aware of an approach based on $\eps$-cuttings where the worst-case bound on the number of curves meeting a connected component is~$o\left(\frac{\abs{\L}}{D}\right)$.

Theorem~\ref{thm:guth} is an extension of the polynomial partitioning theorem by Guth and Katz \cite{GK-15},
based on the polynomial ham-sandwich theorem of Stone and Tukey \cite{ST-42}. Namely,
Guth and Katz showed that, if $\Gamma$ is a finite set of points in ${\reals}^d$ and $D \ge 1$ is an integer parameter,
then there is a non-zero polynomial $f$ of degree at most $D$ so that each connected component of ${\reals}^d \setminus Z(f)$
contains $O(\abs{\Gamma}/D^{d})$ points of $\Gamma$, with a constant of proportionality depending on~$d$. 
Adapting a definition from \cite{AMS-13}, let $r = O(D^{d})$ be an integer parameter with an appropriate constant of proportionality (i.e., it is the number of connected components in ${\reals}^d \setminus Z(f)$ as follows by Warren \cite[Theorem~2]{Warren-68}). We say in this case that $f$ is an \emph{$r$-partitioning polynomial}.
Agarwal, Matou{\v{s}}ek, and Sharir \cite{AMS-13} presented an algorithm that efficiently computes such a polynomial~$f$:\footnote{
  We note that the polynomial $f$ computed in \cite{AMS-13} forms a partition approximating the one shown in \cite{GK-15} in that the constant of proportionality in the degree bound in Theorem~\ref{thm:ams} is slightly worse than that in \cite{GK-15}.}

\begin{theorem}[Effective Polynomial Partitioning for Points \cite{AMS-13}]
  \label{thm:ams}
  Given a set $\mathcal{P}$ of $n$ points in~${\reals}^d$ and an integer parameter~$r \le n$,
  an $r$-partitioning polynomial $f$ for $P$ of degree~$O(r^{1/d})$, with an implicit constant depending on $d$,
  can be computed in randomized expected time~$O(nr + r^3)$.
\end{theorem}

Although the authors in~\cite{AMS-13} do not state so explicitly, the proof of Theorem \ref{thm:ams} also applies to multisets of points, or, equivalently, sets of points with positive integer \emph{weights}, where the weight of a point corresponds to the number of times that it appears in the multiset, where now $n$ denotes the total weight of $P$. This multiset formulation will be useful for our analysis below.

\paragraph*{Model of Computation.}
The algorithm in~\cite{AMS-13} uses the real RAM model of computation, where the input data contains arbitrary real numbers and each arithmetic operation on them is charged unit cost.

In this paper, we additionally assume that, for each $b\geq 1$, the roots of a real univariate polynomial of degree $b$ can be computed in time that depends only on $b$. 
This model was first introduced by Agarwal and Matou{\v{s}}ek~\cite{AM-94} in the context of range searching with semi-algebraic sets (and in fact a variant of it was used as early as in 1983 by Atallah \cite{atallah}), and it has become standard for this type of problems.

\paragraph*{Our Result.}
We present an efficient algorithm that, given a set of algebraic curves in $\reals^3$, partitions~$\reals^3$ into disjoint open ``cells'' (plus a ``boundary'') so that only a small fraction of the curves intersect each cell.
Informally, we prove a theorem of the following kind:

\begin{theorem}[Informal Version]
  \label{thm:cell_decomposition}  
  Let $\L$ be a collection of $n$ irreducible algebraic curves (of constant degree) in $\reals^3$ that satisfy a mild general position requirement\footnote{See Section~\ref{sec:decomposition}.}.
  Let $D$ be a positive integer. Then there is a decomposition of $\reals^3$ into $O(D^{3} \log^3{D})$ disjoint open cells, plus a boundary, so that
  each cell intersects $O(n/D^2)$ curves of $\L$. The \emph{boundary} is the union of an algebraic variety of degree $O(D\log{D})$
  and dimension two, plus an additional semi-algebraic set (with empty interior) that has finite and well-behaved intersection with all but a small number of curves from $\L$.   Moreover, this decomposition can be computed in $O(n^2)$ randomized expected time. For the special case where $\L$ is a set of lines in ${\reals}^3$, the expected running time improves to $O(n^{4/3} \polylog{n})$.
\end{theorem}

A precise statement of Theorem~\ref{thm:cell_decomposition} appears in Section~\ref{sec:decomposition} (Theorem~\ref{thm:cell_decomposition_restated}).
The proof of Theorem~\ref{thm:cell_decomposition}
is based on a two-level decomposition. The first level produces a polynomial partitioning for points
using Theorem~\ref{thm:ams}, and in the second level we apply the method of ``$\eps$-cuttings'' in order to construct an efficient planar decomposition for curve segments, provided that few pairs of curves intersect. This technique also allows us to efficiently partition curve segments in~$\reals^3$,
provided few pairs of curves intersect when projected to the $xy$-plane.
These two ingredients are combined as follows. For each curve in~$\L$, we consider all points on the curve whose projection onto the $xy$-plane lies on the projection of another curve from~$\L$; we call such points ``points of vertical visibility.'' Using Theorem \ref{thm:ams}, we partition $\reals^3$ into connected open cells, so that each cell either intersects
few curves from $\L$, or it contains few points of vertical visibility. Cells of the first kind satisfy the conclusions of our theorem. Cells of the second kind are further decomposed using the $\eps$-cutting machinery mentioned above.

Theorem \ref{thm:cell_decomposition} produces a space decomposition with very similar properties to that of Theorem~\ref{thm:guth} for the case $d = 3$, $k = 1$, though our decomposition is weaker by a polylogarithmic factor. However, because our theorem employs a two-level construction, the ``boundary'' of the cells is not an algebraic variety. Instead, it is the union of an algebraic variety (representing the zero set of an appropriate polynomial obtained at the primary partition) and a semi-algebraic set resulting from the secondary partition. 

\paragraph*{An application: Eliminating depth cycles for lines in $\reals^3$.}

Let $\L$ be a set of $n$ pairwise-disjoint non-vertical lines in $\reals^3$. If $\ell, \ell' \in \L$, we say that $\ell$ passes \emph{above} $\ell'$ if a vertical line that meets both
$\ell$ and $\ell'$ intersects $\ell$ at a point that has larger $z$-coordinate than that of its intersection with $\ell'$;
this line is unique if $\ell$ and $\ell'$ have non-parallel $xy$-projections.
We denote this relation as $\ell' \prec \ell$.
This relation is not necessarily transitive, and is likely to form cycles that consist of three or more lines.

Our goal is to efficiently cut the lines in $\L$ into a finite number of
pieces that do not form any cycles under the relation $\prec$; this is also referred to as \emph{depth order}. In our setting, these resulting pieces are line segments, rays, or just lines in $3$-space. Aronov and Sharir \cite{AS-18} used Theorem~\ref{thm:guth} for the case $d=3$, $k=1$ to obtain a near optimal subquadratic bound on the number of cuts required to create a depth order for lines, provided that the lines are non-vertical and no two lines intersect.
Specifically, they showed that $O(n^{3/2}\log^{O(1)}n)$ cuts suffices. This nearly matches the best-known lower bound $\Omega(n^{3/2})$. We present an efficient implementation of their method, which follows from Theorem \ref{thm:cell_decomposition}.

\begin{theorem}
  \label{thm:depth_cycles}
  Let $\L$ be a collection of $n$ pairwise-disjoint non-vertical lines in~${\reals}^3$. Suppose that no pair of
  projected lines coincide.
  Then, for any prescribed $\eps > 0$, we can cut the lines in $\L$ into $O(n^{3/2 + \eps})$ pieces
  whose depth relation is acyclic. This cutting can be computed in expected $O(n^{3/2+\eps})$ time, where the constant of proportionality depends on $\eps$.
\end{theorem}

The algorithm in Theorem~\ref{thm:depth_cycles} is considerably faster than the algorithms presented in \cite{AS-18} and by De~Berg \cite{deBerg-17}, though the resulting collection of segments is slightly larger (our algorithm produces $(n^{3/2+\eps})$ segments, versus $O(n^{3/2}\log^{O(1)}n)$ segments in \cite{AS-18} and \cite{deBerg-17}), and note that the algorithms in\cite{AS-18, deBerg-17} do not require any general position assumptions on the collection of lines (i.e., line projections may coincide).

The main motivation for eliminating depth cycles comes from hidden surface removal---a technique for rendering a scene in
computer graphics \cite{dBCKO-08}. We refer the reader to the earlier work in~\cite{AKS-05, CEGPSSS-92}, as well as the more recent studies in~\cite{AMS-17, AS-18, deBerg-17} for a comprehensive overview, which also includes the more intricate problem of eliminating depth cycles among pairwise disjoint triangles in $3$-space. In \cite{AMS-17}, the first author, Miller, and Sharir used Theorem~\ref{thm:guth} for the case $d=3$, $k=1$ in order to obtain a near optimal subquadratic bound on the number of pieces required to eliminate depth cycles for triangles in 3-space. In contrast to the case of lines, however, Theorem \ref{thm:cell_decomposition} cannot be used directly to efficiently implement their technique for triangles. This is due to the fact that the structure resulting from Theorem \ref{thm:cell_decomposition} is different than the one resulting from Theorem~\ref{thm:guth}. After the publication of the proceedings version of this paper \cite{AEZ-19}, Agarwal and the authors \cite{AAEZ} have shown, using a different set of tools, that the polynomial partitioning stated in Theorem~\ref{thm:guth} can be computed in an efficient manner. Integrating this result with the mechanism in~\cite{AMS-17} eventually yields an efficient algorithm for the setting of triangles, and, as a result, for the setting of lines (which is considered as a special case of triangles).
We do not discuss this further here.

While the follow-up work~\cite{AAEZ} provides general methods for performing partitioning,  the current work has several advantages:
(i)~It gives a relatively simple solution to the setting of algebraic curves in ${\reals}^3$, and, in particular, bypasses the complicated topological machinery of Guth~\cite{Guth-15} based on the Borsuk-Ulam theorem. So far the only work we are aware of, which bypasses the Borsuk-Ulam theorem, is the one in~\cite{AMS-13} addressing the setting of points in ${\reals}^d$.
(ii)~The running time of the algorithm in~\cite{AAEZ} is linear in the number of input curves, but exponential in the degree $D$ of the polynomial partitioning. On the other hand, the running time of our algorithm is polynomial in~$D$, and therefore it is considerably more efficient when the degree $D$ is non-constant.
(iii)~Last but not least, our work contains novel techniques that give insight into the geometry specific to curve arrangements. There are still many open problems related to curve arrangements, and we hope these ideas will be helpful for future work in this direction. 

\section{Polynomial Partitioning for Algebraic Curves in $3$-Space}
\label{sec:decomposition}
In this section we prove Theorem \ref{thm:cell_decomposition}. Let $\L$ be a collection of $n$~irreducible algebraic curves in~$\reals^3$, each defined by polynomials of degree at most $b$. We will think of $b$ as being fixed, so all implicit constants may depend on $b$. In particular, we write $X(n) = O(Y(n))$ to mean that there exists a constant $C$ depending only on $b$ so that $X(n)\leq C Y(n)$ for all positive integers $n$. We write $X(n) = O_{t}(Y(n))$ to mean that there exists a constant $C$ depending only on $b$ and $t$ so that $X(n)\leq C Y(n)$ for all positive integers $n$.

For a set $X \subset \reals^3$, we denote by $X^*$ its projection onto the $xy$-plane.
Let $\L^*=\{\gamma^* \mid  \gamma\in\mathcal{L}\};$ this is the set of $xy$-projections of the curves of $\L$. 

\paragraph{General position assumptions.}
Let $\L$ be a set of irreducible curves in $\reals^3$. We say that the curves in $\L$ are  \emph{in general position} if the $xy$-projection of each pair of distinct curves from $\L$ have finite intersection. If in addition none of the curves in $\L$ are vertical lines (i.e., lines parallel to the $z$-axis), then we say that the curves in $\L$ are \emph{in non-vertical general position}.

\begin{definition}
  Let $\gamma,\gamma'$ be two distinct irreducible curves in $\reals^3$. A pair of
  points $(p ,p') \in \gamma \times \gamma'$, are called \emph{points of ``vertical visibility'' (with respect to $\gamma$ and $\gamma'$)} if $p^*=(p')^*$.
\end{definition}

If neither $\gamma$ nor $\gamma^\prime$ is a vertical line, and the projections of $\gamma$ and $\gamma^\prime$ have a finite intersection, then $\gamma$ and $\gamma'$ have finitely many pairs of points of vertical visibility.

For the set $\L$ of irreducible curves defined above, we denote by $V(\L)$ the multiset of all points of vertical visibility admitted by the curves in $\L$. If the curves in $\L$ are in non-vertical general position, then each pair of curves from $\L$ contribute $O(1)$ points to $V(\L)$, and thus $|V(\L)|=O(n^2)$.

Our main algorithm assumes that the curves in $\L$ are in non-vertical general position. Later in the paper we show how this assumption can be weakened to merely assume that the curves are in general position.

\begin{remark}
  We do not make any further general position assumptions.
  In particular, the curves in $\L$ are allowed to intersect, and the $xy$-projections of three or more curves from $\L$ are allowed to intersect at a point. If a point $p\in V(\L)$ is contained in $m$ curves from $\L$ and its $xy$-projection $p^*$ is contained in the $xy$-projection of $k\geq m$ curves from $\L$, then $p$ will have multiplicity $m(k-1)$ in $V(\L)$, i.e., it will have weight $m(k-1)$. This observation will be exploited when we apply Theorem \ref{thm:ams}.
  Note also that a single curve $\gamma \in \L$ may intersect itself.
  Moreover, a vertical line may intersect $\gamma$ at several points, implying that the $xy$-projection of $\gamma$ intersects itself.  However, we do not view these points as points of vertical visibility. 
  In fact, this self-intersection might be a 
  semi-algebraic set of positive dimension (this will occur, for example, if $\gamma$ is a circle whose projection to the $xy$-plane is a line segment), in which case there is an overlap in the $xy$-projection of $\gamma$.
  We revisit this scenario in Section~\ref{sec:second_step}, where we describe how to incorporate that into our analysis.
\end{remark}

We will need a number of tools to analyze the collection of curves $\L$. In \cite{BB-12}, Barone and Basu used considerations from real algebraic geometry in order to obtain several analogues of B\'ezout's theorem and Harnack's curve theorem, which are useful when studying algebraic space curves. We will state two special cases of their results here. 

\begin{theorem}[B\'ezout's theorem for real algebraic space curves]\label{thm:Bezout}
For each integer $b\geq 1$, there is a constant $C_1(b)$ so that the following holds. Let $\gamma$ be an irreducible algebraic curve in $\reals^3$ that is defined by polynomials of degree at most $b$. Let $P$ be a trivariate polynomial of degree $D$. Then either $P$ vanishes on $\gamma$ (i.e., $\gamma\subset Z(P)$) or $|\gamma\cap Z(P)|\leq C_1(b)D$. 
\end{theorem}

\begin{theorem}[Harnack's curve theorem for real algebraic space curves]\label{thm:Harnack}
For each integer $b\geq 1$, there is a constant $C_2(b)$ so that the following holds. Let $\gamma$ be an irreducible algebraic curve in $\reals^3$ that is defined by polynomials of degree at most $b$. Then $\gamma$ is a union of at most $C_2(b)$ connected components. The same result holds if $\gamma\subset\reals^2$ is an irreducible algebraic plane curve. 
\end{theorem}

We remark that the constants $C_1(b)$ and $C_2(b)$ can be computed explicitly, but this is not required by our analysis.

Hereafter we fix a parameter $D \geq 1$, which will play a role analogous to the parameter $D$ from Theorem \ref{thm:guth}. As a minor technicality, we would like to assume that $D$ is not much larger than $n^{1/2}$. The following lemma allows us to dispense with the case when $D$ is large compared to $n^{1/2}$

\begin{lemma}
For each integer $b\geq 1$, there is a constant $C_3(b)$ so that the following holds. Let $\L$ be a collection of $n$~irreducible algebraic curves in~$\reals^3$, each defined by polynomials of degree at most $b$, and let $D>C_3(b)n^{1/2}$. Then there is a polynomial $P$ of degree at most~$D$ that vanishes on each curve in $\L$. This polynomial can be constructed in $O(\poly(D))$ time.
\end{lemma}
\begin{proof}
Select $C_1(b) D + 1$ points on each curve from $\L$, where $C_1(b)$ is the constant from Theorem \ref{thm:Bezout}, and let $P$ be a polynomial of degree~$D$ that vanishes at these points. Since the vector space of polynomials of degree at most~$D$ has dimension~$\binom{D+3}{3}$, if $C_3(b)$ is sufficiently large (depending only on $C_1(b)$, which in turn depends only on $b$) then the condition $D>C_3(b)n^{1/2}$ implies that~$\binom{D+3}{3}> n(C_1(b) D+1)$, so there must exist a non-zero polynomial that vanishes at the specified points. The coefficients of such a polynomial can be computed in $O(\poly(D))$ time (using, e.g., Gaussian elimination). Since $P$ vanishes on at least $C_1(b) D + 1$ points from each curve in $\L$, Theorem \ref{thm:Bezout} implies that $P$ vanishes on each curve from $\L$.
\end{proof}

Henceforth we will assume that $D=O(n^{1/2})$. Our space decomposition is constructed by two main partitioning steps.
In the first, we iteratively partition space by overlaying the zero sets of polynomials of degree $D \ge 1$,
each of which partitions a subset of $V(\L)$, so that the overall majority of resulting cells meet only
$O(n/D^2)$ curves of $\L$ each, and the remaining cells together cover a small fraction of~$V(\L)$
(Lemma~\ref{lem:firstPartitioning}).  By applying this process $O(\log{D})$ times, we obtain a trivariate
polynomial~$P$ of degree $O(D\log D)$, which partitions space into $O((D\log D)^3)$ open connected cells, each of which either
intersects only $O(n/D^2)$ curves from~$\L$, or contains $O(n^2/D^4)$ points of $V(\L)$ (or both). Cells that intersect $O(n/D^2)$ curves from~$\L$ are called \emph{acceptable}, while the remaining cells are deemed \emph{unacceptable}. This step is performed in Corollary~\ref{cor:firstPartitioning}. 

An unacceptable cell intersects a large number of curves, but the fact that it contains $O(n^2/D^4)$ points from $V(\L)$ allows us to further decompose it into a small number of acceptable subcells. This leads to the second decomposition step, in which we build an \emph{$\eps$-cutting} within each such cell (the $\eps$-cutting is constructed in the $xy$-projection of the cell, but then we apply a lifting in the $z$-direction in order to obtain a three-dimensional decomposition). This is described in Section~\ref{sec:second_step}. The construction is based on the random sampling technique of Clarkson and Shor \cite{CS-89} and of Chazelle and Friedman \cite{CF-90}, which was later adapted by De~Berg and Schwarzkopf \cite{dBS-95} to yield output-sensitive decompositions.

\subsection{The First Decomposition Step: Iteratively Partition Space}
\label{sec:first_step}
We first show the following main property.

\begin{lemma}
  \label{lem:firstPartitioning}
    For each integer $b\geq 1$, there is a constant $C_4(b)$ so that the following holds. Let $\L$ be a collection of $n$ irreducible algebraic curves in~$\reals^3$ in non-vertical general position, each defined by polynomials of degree at most~$b$.  Let $D$ be a positive integer.
  For each non-negative integer $k \geq 0$, there is a set $V_k \subset V(\L)$ and a polynomial $F_k$ with the following properties:
  \begin{enumerate}[(A)]
  \item $\deg(F_k) \leq kD$. 
  \item $\abs{V_k} \leq \abs{V(\L)}/2^k$.
  \item For each open connected component $\Omega$ of ${\reals}^3\setminus Z(F_k)$, at least one of the following holds:
  \begin{enumerate}
  \item[(C.1)] $\Omega$ intersects at most $C_4(b) n/D^2$ curves from $\L$, or
  \item[(C.2)] $\Omega \cap V(\L) \subset V_k$.
  \end{enumerate}
  \end{enumerate}
\end{lemma}

\begin{proof}
  First note that the curves from $\L$ fully contained in $Z(F_k)$ can be disregarded, 
  since they do not meet any connected components of $\reals^3\setminus Z(F_k)$, and are therefore irrelevant
  for the assertions of the lemma.
  
  We prove properties (A)--(C) by induction on $k$. For $k=0$, the assertions are satisfied by putting $V_0 = V(\L)$ and $F_0 = 1$.
  For $k \ge 1$, let $V_{k-1}$ be a set of points and let $F_{k-1}$ be a polynomial satisfying properties~(A)--(C) above (with $k-1$).
  Apply Theorem \ref{thm:ams} to find a partitioning polynomial~$f$ of degree at most~$D$ for the multiset of points $V_{k-1}$.
  Each connected component of ${\reals}^3 \setminus Z(f)$ contains $O(\abs{V_{k-1} } /D^3)$ points from~$V_{k-1}$.

  We call a connected component of ${\reals}^3 \setminus Z(f)$ an \emph{acceptable} cell if it intersects at most $C_4(b) n / D^2$ curves from $\L$ (we specify the choice of $C_4(b)$ shortly); otherwise we call it an \emph{unacceptable} cell.

By Theorem \ref{thm:Harnack}, each curve $\gamma \in \L$ contains at most $C_2(b)$ irreducible components. If $\gamma$ is not contained in $Z(f)$ then by Theorem \ref{thm:Bezout}, we have $|\gamma\cap Z(f)|\leq C_1(b)D+1$. This implies that each curve $\gamma \in \L$ that is not contained in $Z(f)$ intersects at most $C_1(b)D+C_2(b)$ connected component from ${\reals}^3 \setminus Z(f)$, and thus there are at most $n(C_1(b)D+C_2(b))$ pairs $(\gamma,\Omega)$, where $\gamma\in\L$ is a curve, $\Omega$ is a connected component from ${\reals}^3 \setminus Z(f)$, and $\gamma$ intersects $\Omega$. Since each unacceptable cell accounts for at least $C_4(b)n/D^2$ of these pairs, at most 
$\frac{n(C_1(b)D+C_2(b))}{C_4(b)n/D^2}\leq \frac{C_1(b)+C_2(b)}{C_4(b)}D^3$ cells are unacceptable. Since each unacceptable cell contains $O(\abs{V_{k-1} } /D^3)$ points from $V_{k-1}$, if we select the constant $C_4(b)$ sufficiently large (depending only on $C_1(b), C_2(b)$ and thus on $b$), 
then at most $\abs{V_{k-1}} / 2$ points from $V_{k-1}$ are contained in unacceptable cells.

Define
\begin{equation}
  \label{eq:v_k}
  V_k \coloneqq \bigcup_\text{$\tau$ unacceptable}
  \tau \cap V_{k-1} ,
\end{equation}
with the union taken over all unacceptable cells~$\tau$ of~$\reals^3\setminus Z(f)$.
To complete the inductive step, we define $F_k \coloneqq F_{k-1} \cdot f$. In other words,
$Z(F_k) = Z(F_{k-1}) \cup Z(f)$. 

Then 
\[
\deg(F_k) = \deg(F_{k-1}) + \deg(f) \le (k-1)D + D = kD,
\] 
so property~(A) is satisfied. We have $\abs{V_k} \leq \abs{V_{k-1}}/2 \leq \abs{V_0}/2^k$, thus property~(B) is satisfied. It remains to verify property~(C). Let $\Omega$ be a connected component of ${\reals}^3\setminus Z(F_k)$; this component is contained in the intersection of some connected component $\Omega'$ of $\reals^3\setminus Z(F_{k-1})$ and a connected component $\tau$ of $\reals^3\setminus Z(f)$. If $\Omega'$ intersects at most $C_4(b) n/D^2$ curves from~$\L$, so does $\Omega$, and (C.1)~holds. On the other hand, if $\Omega'$ intersects more than $C_4(b) n/D^2$ curves from~$\L$, then by property (C.2) of the induction hypothesis, $\Omega' \cap V(\L) \subset V_{k-1}$, and, in fact, $\Omega' \cap V(\L) = \Omega' \cap V_{k-1}$.
If $\Omega$ is unacceptable, then so is $\tau \supseteq \Omega$ and \eqref{eq:v_k}~implies
\[
\Omega \cap V_k = \Omega \cap V_{k-1} = \Omega \cap V(\L) ,
\]
since $\Omega \subset \Omega'$.  Therefore (C.2) holds.
Otherwise, $\Omega$ is acceptable, so it intersects at most $C_4(b) n/D^2$~curves from $\L$ and (C.1)~holds.
Thus property~(C) is satisfied, which concludes
the inductive argument.
\end{proof}

\noindent{\bf Remark:}
We note that, since we apply Theorem \ref{thm:ams} to a multiset of points, it in particular implies that a point with a high multiplicity is more likely to appear on the zero set of the partitioning polynomial.

Applying Lemma \ref{lem:firstPartitioning} to $\L$ with parameter $D$ and $k = 4\lceil\log_2 D\rceil$, and recalling that our implicit constants are allowed to depend on $b$, so in particular $C_4(b) = O(1)$, we conclude: 
\begin{corollary}
  \label{cor:firstPartitioning}
  Let $\L$ be a collection of $n$ irreducible algebraic curves in~$\reals^3$ in non-vertical general position, each defined by polynomials of degree at most~$b$.
  Let $D$ be a positive integer. Then there is a polynomial $P \in \reals[x,y,z]$ of degree $O(D\log D)$ such that ${\reals}^3 \setminus Z(P)$ is a union of $O((D\log D)^3)$ open connected components, which we will call the \emph{cells} of the decomposition, so that for each such cell $\Omega$, at least one of the following holds: either $\Omega$ intersects $O(n/D^2)$ curves from $\L$, or $\Omega$ contains $O(n^2/D^4)$ points of $V(\L)$ (or both).
\end{corollary}

\noindent{\bf Remark:}
Although at this point the constant $4$ in the choice of $k$ seems arbitrary, it will become clearer in Section~\ref{sec:second_step}.   In particular, this choice is exploited in the proof of Claim~\ref{clm:subcells}.

Next, we will weaken the requirement that the curves in $\L$ are in non-vertical general position, and replace it with the requirement that the curves are in general position.

\begin{corollary}
  \label{cor:firstPartitioningWithVert}
  Let $\L$ be a collection of $n$ irreducible algebraic curves in~$\reals^3$ in general position, each defined by polynomials of degree at most~$b$. Write $\L=\L_1\cup\L_2$, where the curves in $\L_1$ are vertical lines and none of the curves in $\L_2$ are vertical lines. By assumption, the $xy$-projections of each pair of curves from~$\L_2$ have finite intersection.
  Let $D$ be a positive integer. Then there is a polynomial $P \in \reals[x,y,z]$ of degree $O(D\log D)$ such that ${\reals}^3 \setminus Z(P)$ is a union of $O((D\log D)^3)$ open connected components, which we will call the \emph{cells} of the decomposition, where each cell~$\Omega$ intersects~$O(n/D^2)$ curves from $\L_1$. Furthermore, for each such cell $\Omega$, at least one of the following holds: either $\Omega$ intersects $O(n/D^2)$ curves from $\L_2$, or $\Omega$ must contain at most $O(n^2/D^4)$ points of $V(\L_2)$ (or both).
\end{corollary}

\begin{proof}
Let $\L_1^*\subset\reals^2$ be the set of points obtained by intersecting the curves in $\L_1$ with the $xy$-plane. Apply Theorem \ref{thm:ams} to find a partitioning polynomial~$P_1$ of degree at most~$D$ for~$\L_1^{*}$. Each connected component of ${\reals}^2 \setminus Z(P_1)$ contains $O(\abs{\L_1^{*}} /D^2)$ points from~$\L_1^{*}$. We now consider $P_1$ as a polynomial~$P_1(x,y,z)$ of three variables ($P_1(x,y,z)$ is independent of $z$, so $Z(P_1(x,y,z))$ is obtained by lifting~$Z(P(x,y))$ in the  $z$-direction). 
Let $P_2$ be the output of Corollary \ref{cor:firstPartitioning} applied to $\L_2$, and define~$P \coloneqq P_1P_2$.
\end{proof}

\subsection{The Second Decomposition Step: Random Sampling}
\label{sec:second_step}

In this section we show how to further refine the decomposition obtained in Corollary~\ref{cor:firstPartitioningWithVert}, so
that all cells are acceptable. 

Write $\L=\L_1\cup\L_2$, as in the statement of Corollary~\ref{cor:firstPartitioningWithVert}. Fix an unacceptable cell $\Omega \in {\reals}^3 \setminus Z(P)$. From Corollary~\ref{cor:firstPartitioning}, it follows that
$\Omega$ contains $O(n^2/D^4)$ points of $V(\L_2)$ (counting with multiplicity).
Let $\L_{\Omega} \subseteq \L_2$ be the subset of curves that meet $\Omega$.
We now intersect each curve $\gamma \in \L_{\Omega}$ with $\Omega$.
Let $S_{\Omega}$ be the collection of the resulting open curve segments lying in $\Omega$,
and let $S^{*}_{\Omega}$ be the set of their projections onto the $xy$-plane.
Recall that we allow the curves in $\L_2$ to self-intersect.  Moreover, a vertical line might intersect a curve of $\L_2$ at several points. 
This implies that the projected curves in $S^{*}_{\Omega}$ may form self-intersections, of which we dispose as follows. Each projected curve $\gamma\in S^{*}_{\Omega}$ is contained in the zero set of a square-free polynomial $g_\gamma$ of degree $O(1)$, where the implicit constant depends only on $b$. By Theorem \ref{thm:Harnack}, $Z(g)$ contains $O(1)$ connected components. By \cite[Lemma 2.3]{SZ-17} and B\'ezout's theorem, $Z(g)\backslash Z(\partial_y g)$ is a union of $O(1)$ $x$-monotone (open) Jordan arcs, where the implicit constant depends only on $b$. Let $W_\Omega$ be the set of such arcs; we have $\abs{W_\Omega} = O(\abs{S^{*}_{\Omega}})$, where again, the constant of proportionality depends on the degree $b$.

\paragraph{Planar arrangements and $(1/r)$-cuttings.}

We will recall a few standard definitions about arrangements of curves; see \cite{AS_handbook-00} for further details. Given a set $W$ of algebraic arcs in the plane, the \emph{arrangement} $\A(W)$ of $W$ is the partition of the plane induced
by the arcs in $W$ into vertices, edges, and faces, where a \emph{vertex} is either an endpoint of an arc or the intersection point
of a pair of arcs, an \emph{edge} is a (relatively open) portion of an arc delimited by two consecutive vertices, and a \emph{face}
is a maximal connected open planar region that is disjoint from the arcs and vertices in $W$. We say that a face \emph{meets} an edge from $\A(W)$ if the face intersects the edge.
The total complexity of the arrangement $\A(W)$ is the overall number of its vertices, edges, and faces.
The \emph{vertical decomposition} of $\A(W)$ is a partition of the faces of $\A(W)$ into pseudo-trapezoidal faces, by erecting upward and downward vertical walls from each vertex of $\A(W)$ until it hits the first vertex or edge of $\A(W)$, or continuing to $\infty$ (or $-\infty$) if there is no such edge or vertex. 

Let $r > 0$ be a real parameter to be fixed shortly. Our goal is to construct a \emph{$(1/r)$-cutting} 
for the arcs in $W_\Omega$.  This is a partition of the plane into constant-complexity faces (in our case
these are pseudo-trapezoidal faces determined by the vertical decomposition of the two-dimensional arrangement of some arcs from~$W_{\Omega}$ \cite{AS_handbook-00}), each of which meets at most $\abs{W_\Omega}/r$ of the arcs in $W_{\Omega}$.
From \cite[Lemma~2.2]{dBS-95} it follows that there is a $(1/r)$-cutting for $W_{\Omega}$ consisting
of $O(\tau(r))$ pseudo-trapezoidal faces. Here $\tau(r)$ is the expected number of faces in the vertical decomposition\footnote{The original formulation in \cite{dBS-95} is for canonical triangulations, but in our case, they can be replaced with vertical decompositions.}
$\VD(R^{*}_{\Omega})$ of the arrangement $\A(R^{*}_{\Omega})$ of a random subset $R^{*}_{\Omega} \subseteq  W_{\Omega}$, where every arc in $W_{\Omega}$ is drawn independently with probability $p \coloneqq \frac{c D^2}{n}$ for a fixed 
constant $c > 0$. We set $r \coloneqq p \abs{W_{\Omega}}$.
The expected number of pseudo-trapezoidal faces in $\VD(R^{*}_{\Omega})$ is proportional to the expected complexity of the arrangement $\A(R^{*}_{\Omega})$ \cite{AS_handbook-00}). 
We next show the following claim.

\begin{claim}
  \label{clm:subcells}
  Letting $m_{\Omega} \coloneqq \EE[\abs{R^{*}_{\Omega}}]$,
  the expected complexity of $\A(R^{*}_{\Omega})$ is $O(m_{\Omega} + 1)$,
  where $\EE[\cdot]$ denotes expectation.
\end{claim}
\begin{proof}
  
  Let $X$ be the set of vertices of $\A(W_\Omega).$ For each vertex $x\in X$, let $w(x)$ be the weight of $x$, that is, the number of pairs of curves from the arrangement $\A(W_\Omega)$ that contain $x$. Since $\Omega$ contains $O(n^2/D^4)$ points of $V(\L)$ (counting with multiplicity), we have that $\sum_{x\in X}w(x) = O(n^2/D^4)$, and thus in particular\footnote{The number of vertical visibilities might be considerably smaller, as such a visibility is relevant for a pair of curve segments $\gamma_1, \gamma_2 \in S_{\Omega}$ only if both points $v_1 \in \gamma_1$ and $v_2 \in \gamma_2$ (which lie vertically above the other) are contained in $\Omega$, but it may happen that only one of $v_1, v_2$ is in $\Omega$.} $|X|=O(n^2/D^4)$.

  Let $R^{*}_{\Omega}$ be a random subset of $W_\Omega$, where every arc in $W_{\Omega}$ is drawn independently with probability $p \coloneqq \frac{c D^2}{n}$. Since by our earlier assumption $D=O(n^{1/2})$, if $c$ is chosen sufficiently small (depending on $b$), then we can ensure that $p<1/2$. For each $x\in X$, let $w^*(x)$~be the weight of the vertex $x$ in this random set. Since $p<1$, we have
\[
  \EE[w^*(x)]=\sum_{k=2}^{w(x)}\binom{k}{2} p^k < \sum_{k=2}^\infty \binom{k}{2}p^k = \frac{p^2}{(1-p)^3}.
\]
By linearity of expectation, the expected combined weight of the vertices in $\A(R^{*}_{\Omega})$ is
\[
  \EE\Big[\sum_{x\in X}w^*(x)\Big]=\sum_{x\in X}\EE[w^*(x)]< \sum_{x\in X} \frac{p^2}{(1-p)^3} =  \frac{p^2}{(1-p)^3}|X|.
\]
Having $p < 1/2$ as above, we obtain $|X| \frac{p^2}{(1-p)^3}\leq 8p^2|X|=O(1)$. The claim now follows from the fact that the arrangement complexity of $\A(R^{*}_{\Omega})$ is bounded (up to multiplicative constants) by the number of elements in $R^{*}_{\Omega}$ plus the number of intersections between pairs of curves in the corresponding arrangement.
\end{proof}

We next bound the total expected complexity of $\A(R^{*}_{\Omega})$ (and thus the total number of faces in~$\VD(R^{*}_{\Omega})$),
over all unacceptable cells $\Omega \in {\reals}^3 \setminus Z(P)$.
Put \[
  W \coloneqq \bigcup_{\text{$\Omega$ unacceptable}}
  W_{\Omega} ,
\]
with the disjoint union taken over all unacceptable cells $\Omega$ of $\reals^3 \setminus Z(P)$,
and recall that within each such cell $\Omega$ an arc of $W_{\Omega}$ is selected independently with probability $p = \frac{c D^2}{n}$.
Since $\deg(P) = O(D \log{D})$, we have $\abs{W} = O(n D\log{D} + n)$, as we started with $n$ curves and cut them into pieces at the $O(n D\log{D})$ points of intersection with $Z(P)$.  We also recall that we cut the $xy$-projections of these pieces into $x$-monotone Jordan arcs (which only increases the number of arcs by a constant factor that depends on $b$).
Therefore the total expected number of arcs in the samples $R^{*}_{\Omega}$, over all unacceptable
cells $\Omega$, is $O(p n D \log{D}) = O(D^3 \log{D})$.

Claim~\ref{clm:subcells} now implies that
\begin{equation}
  \label{eq:second_stage_cells}
  \EE\Bigl[ \sum_\Omega
  ( \abs{R^{*}_{\Omega}} + 1) \Bigr] =
  O(D^3 \log^3{D}) ,
\end{equation}
with the summation taken over all unacceptable cells $\Omega$ of $\reals^3 \setminus Z(P)$.
In other words,  we have just shown that the expected total number of faces in $\VD(R^{*}_{\Omega})$ over all such cells $\Omega$
is $O(D^3 \log^3{D})$.

We finally describe the actual refinement of the unacceptable cells~$\Omega$.
Each pseudo-trapezoidal face~$\Delta^{*}_{\Omega} \in \VD(R^{*}_{\Omega})$ is turned into a vertical prism~$\sigma$ by taking its Cartesian product with the $z$-axis.
We now form intersections of $\Omega$ with each prism $\sigma$; $\Omega$ is only intersected with prisms arising from the pseudo-trapezoidal faces of its own decomposition $\VD(R^{*}_{\Omega})$.
We refer to these intersections as the (open) \emph{second-stage cells} and observe that they might not be connected, since $\Omega$ needs not be $xy$-monotone. Despite this oddity, our decomposition does have the desired properties. 

Indeed, since each second-stage cell $\xi=\Omega \cap \sigma$ corresponds to a unique pseudo-trapezoid~$\Delta^{*}_{\Omega}$, the overall expected number of second-stage cells is $O(D^3 \log^3{D})$.
By the properties of $(1/r)$-cuttings, each pseudo-trapezoidal face $\Delta^{*}_{\Omega} \in \VD(R^{*}_{\Omega})$
meets $|W_{\Omega}|/r$ arcs of $W_{\Omega}$. Since $r = \frac{c D^2}{n} \cdot \abs{W_{\Omega}}$, 
each $\Delta^{*}_{\Omega}$ meets $O(n/D^2)$ arcs of $W_{\Omega}$.
Therefore $\Delta^{*}_{\Omega}$ meets $O(n/D^2)$ curves of $\L_2^{*}$.\footnote{This is potentially an overestimate, since $\Delta^{*}_{\Omega}$ may meet several arcs of the same original curve of $\L_2$.} So the number of curves from $\L_2$ met by $\sigma$, and, in particular, by the actual cell $\xi = \Omega \cap \sigma$ is $O(n/D^2)$, as claimed. Since $\xi$ is a subset of a cell from Corollary \ref{cor:firstPartitioningWithVert}, we have that $\xi$ intersects $O(n/D^2)$ curves from $\L_1$. Thus $\xi$ intersects $O(n/D^2)$ curves from $\L$. 

Recall that, by Corollary~\ref{cor:firstPartitioningWithVert}, the number of the remaining (that is, acceptable) cells in ${\reals}^3 \setminus Z(P)$ is $O(D^3 \log^3{D})$, and each of these cells meets $O(n/D^2)$ curves from $\L$.
To summarize, in both levels of the decomposition we obtain $O(D^3 \log^3{D})$ cells in total, each meeting $O(n/D^2)$ curves of $\L$.

\paragraph{Wrapping up.}
We claim that the cell decomposition described above satisfies the properties stated in Theorem~\ref{thm:cell_decomposition}. We state these properties more formally below.

\begin{theorem}[Theorem \ref{thm:cell_decomposition} restated]
  \label{thm:cell_decomposition_restated} 
  Let $\L$ be a collection of $n$ irreducible algebraic curves in~$\reals^3$, each defined by polynomials of degree at most $b$. Suppose that the projections of each pair of curves from $\L$ to the $xy$-plane have finite intersection.  Let $D$ be a positive integer. Then there is a number $N = O(D^3\log^3 D)$ and a partition $\reals^3 = Z\cup \bigcup_{i=1}^N K_i$ (into a boundary~$Z$ and cells~$K_i$)
  with the following properties. 
  \begin{itemize}
  \item Each $K_i$ is an open (not necessarily connected) cell, consisting of a union of connected components of $\reals^3\setminus Z$.
  \item Each such cell intersects $O(n/D^2)$ curves from $\L$.
  \item The interior of $Z$ is empty, and there is a trivariate polynomial $P$ of degree $O(D\log D)$, with $Z(P)\subset Z$.
  \item The curves from $\L$ not contained in $Z(P)$ intersect $Z$ in relatively few points, excluding a subset $\L' \subset \L$ of $O(D^3 \log{D})$ curves.
    Specifically,
    \begin{equation}
      \label{numberOfCurveBoundaryIntersections}
      \sum_{\substack{\gamma \in \L \setminus \L' \\ \gamma \not\subset Z(P) }} \abs{\gamma\cap Z} = O(n D\log^3 D).
    \end{equation}
  \end{itemize}
  This partition can be computed in $O_D(n^2)$ randomized expected time, where the algorithm outputs for each cell $K_i$ the list of curves from $\L$ that it intersects.
  For the special case where $\L$ is a set of lines in $3$-space that satisfy a mild general position requirement\footnote{
    See Section~\ref{sec:lines_alg}.}
  the expected running time improves to $O_D(n^{4/3} \polylog{n})$.
\end{theorem}

The analysis of $(1/r)$-cuttings in \cite{dBS-95} guarantees that there exists a choice of the random samples~$R^{*}_{\Omega}$, such that each of the unacceptable cells has been subdivided into subcells that intersect $O(n/D^2)$ curves from $\L$. If a curve $\gamma\in\L$ is not contained in $Z(P)$, then it intersects $Z(P)$ in $O(D\log D)$ points. Thus the total number of intersections between curves in $\L$ not contained in~$Z(P)$ and~$Z(P)$ is $O(n D\log D)$.  A curve $\gamma$ (not contained in $Z(P)$) intersects a vertical wall of a second-level cell in $O(1)$ points, if $\gamma^{*}$ does not have any curve segments comprising the sets $W_{\Omega}$ defined above, which participate in the samples $R^{*}_{\Omega}$.
Otherwise, $\gamma$ intersects some of the vertical walls of the second-level cells constructed within $\Omega$ in a curve segment.
Curves $\gamma$ of the latter kind comprise the set $\L' \subset \L$, and, as argued above, their total expected number is $O(D^3 \log{D})$. 
Thus the total number of intersections between curves from $\L \setminus \L'$ (not contained in $Z(P)$) and cells, resulting either in the first or second stage of the decomposition, is $O( (D\log D)^3 (n/D^2)) = O(n D\log^3D)$.
This establishes \eqref{numberOfCurveBoundaryIntersections}.

The implementation details concerning the expected running time of the algorithm, as stated in Theorem~\ref{thm:cell_decomposition_restated}, are given in Section~\ref{sec:alg_aspects} below.

\begin{remark}
  In higher dimensions, the analogue of Theorem \ref{thm:cell_decomposition_restated} is a decomposition of $\reals^d$ into $O( (D\log D)^d)$ cells, each of which intersects $O(n/D^{d-1})$ curves from $\L$. We remark that most of the steps from the proof of Theorem \ref{thm:cell_decomposition_restated} extend to higher dimensions. The main difficulty in this extension is handling a curve from $\L$ that projects to a point in the $xy$-plane.  Unlike the three-dimensional case, such a curve is not necessarily ``vertical'' (i.e., parallel to an axis), and therefore
  the solution from 
  Corollary~\ref{cor:firstPartitioningWithVert} does not apply.

  One simple way to resolve this issue is to impose an additional general position assumption on the curves in $\L$. For example, we could require that for each curve $\gamma\in\L$, any fiber of the projection map $\pi \colon \gamma\to\reals^2$ from $\gamma$ to the $xy$-plane has finite cardinality. With this additional assumption, the proof of Theorem~\ref{thm:cell_decomposition_restated} extends with only minor modifications. Indeed, Lemma~\ref{lem:firstPartitioning} extends to $\reals^d$ by applying it with $k=(2d-2)\lceil\log_2 D\rceil$, which yields an analogue of Corollary~\ref{cor:firstPartitioning}, where each acceptable cell intersects $O(n/D^{d-1})$ curves from $\L$, and each unacceptable cell contains $O(n^2 / D^{2d-2})$ points of vertical visibility. The analysis of the second-level cell decomposition proceeds almost verbatim where we produce  a planar vertical decomposition. In this extension the sampling probability $p$ becomes $cD^{n-1}/n$ (we can assume that $D=O(n^{1/(d-1)})$ and thus $p<1/2$, since otherwise there exists a polynomial of degree $O(D)$ whose zero set contains all of the curves in $\L$), and the planar pseudo-trapezoids (once constructed) are lifted to prisms in $\reals^d$ in all $d-2$ residual coordinates.
\end{remark}

\subsection{Algorithmic Aspects}
\label{sec:alg_aspects}

We now outline how to efficiently implement Theorem~\ref{thm:cell_decomposition_restated}. This involves identifying curves that are non-vertical lines and constructing $V(\L)$ (which can be accomplished in quadratic time by brute-force examination of all pairs of curves), followed by several rounds of invocation of Theorem~\ref{thm:ams}; the random sampling step and the construction of vertical decomposition in the plane are standard and are not the bottleneck of the algorithm.  
A major technicality arises from the fact that we need to keep track of the points of~$V(\L)$ contained in each cell, to determine which cells are acceptable.

Recall that, as described in the introduction, we are using the real RAM model of computation, with the additional assumption that the roots of a univariate polynomial of degree $b$ can be computed in $O_b(1)$ time. The other common alternative is to assume integer coefficients and express the running time in the number of bit operations, as a function of the size of input and the bit length of the input coefficients; see~\cite[Chapter 8]{BPR-06} for a discussion of such models of computation.  We chose to proceed with the former model.

As a main tool, we use a result of Basu, Pollack, and Roy \cite[Algorithm 16.6]{BPR-06} concerning arrangements
of zero sets of polynomials (see also \cite[Theorem 4.1]{AMS-13} for a similar formulation):

\begin{theorem}[Basu, Pollack, and Roy \cite{BPR-06}]
  \label{thm:connected_components}
  Let $\F = \{f_1, \ldots, f_s \}$ be a set of $s$ real $d$-variate polynomials, each of degree at most $D$.
  Then the arrangement $A(\F)$ in ${\reals}^d$ has $O(1)^d (sD)^d$ connected components,
  and it can be computed in time at most $T = s^{d+1} D^{O(d^4)}$.
  Each connected component is described as a semi-algebraic set using at most $T$ polynomials of degree bounded by $D^{O(d^3)}$.
\end{theorem}

Following the notation of Lemma~\ref{lem:firstPartitioning}, at each step $k > 1$, we are given a subset $V_{k-1} \subseteq V(L)$ and a previously computed partitioning polynomial $F_{k-1}$. We first apply Theorem \ref{thm:ams} to compute a partitioning polynomial $f$ of degree $D$ for the point set $V_{k-1}$. Then, to determine if a cell~$\tau$ in~${\reals}^3 \setminus Z(f)$ is acceptable, we test whether $\tau$ intersects at most $C_4(b) n / D^2$ curves from $\L$. 
By applying Theorem~\ref{thm:connected_components} to our polynomial $f$, 
we represent each cell $\tau \in {\reals}^3 \setminus Z(f)$ as a semi-algebraic set (a Boolean formula with polynomial sign tests as atoms),
and then test, for each curve $\gamma \in \L$, 
whether $\gamma$ intersects $\tau$. This can be done in two steps. First, we use Theorem \ref{thm:connected_components} to compute all of the connected components of $\gamma\backslash  Z(f)$ (recall that $\gamma$ is defined by polynomials of degree at most $b$, and this can in fact be replaced by a single polynomial of degree at most $2b$), and we assign each of these components to the corresponding cell $\tau$ (the latter can be done by computing a point on each of the connected components, and checking which cell contains this point). Overall, this can be done in time $D^{O(1)}$ (see also \cite{AMS-13} for similar considerations).
Thus the total running time, over all curves in $\L$, is $n O_D(1) + n D^{O(1)} = O_D(n)$.
Next, we compute the new polynomial $F_k = f \cdot F_{k-1}$, and then form the subset $V_k$, by testing for each point of $V_{k-1}$ whether it lies in an unacceptable cell, using the membership test already discussed.
Applying Theorem~\ref{thm:connected_components} and the fact that
$\abs{V_{k-1}} \le \abs{V(\L)} = O(n^2)$, we can complete this task in time $O_D(n^2)$.
The process repeats $O(\log D)$ times and takes $O_D(n^2)$ time in total.

We next need to construct a decomposition into vertical prisms within each unacceptable cell~$\Omega \in {\reals}^3 \setminus Z(P)$.
This involves the computation of $(1/r)$-cutting within $\Omega$. Recall that we apply the randomized algorithm described
in \cite[Theorem 2.1]{dBS-95}, which constructs a $(1/r)$-cutting in an
arrangement of $m$ arcs in expected time
$O(m\log{r} + A \cdot r/m)$, where $A$ is the total number of intersections among these arcs. Applying this to each of the arrangements $\A(W_\Omega)$, by substituting $m := \abs{W_{\Omega}}$, $A := \abs{\V(W_{\Omega})}$ (where $\abs{\V(W_{\Omega})}$ is the number of vertices in the underlying arrangement), we obtain an expected running time of 
$O\left( \abs{W_{\Omega}}\log{r} +\abs{\V(W_{\Omega})} \cdot \frac{r}{\abs{W_{\Omega}}}\right)$.
Recall that
$\abs{\V(W_{\Omega})} = O(\abs{W_{\Omega}} + n^2/D^4)$.
Then
it is easy to verify that the expected running time, over all unacceptable cells $\Omega$, is proportional to
\begin{align*}
  \sum_\Omega \left(\abs{W_{\Omega}}\log{r} +\abs{\V(W_{\Omega})} \cdot \frac{r}{\abs{W_{\Omega}}} \right) & =
\sum_\Omega \left((\abs{W_{\Omega}}\log{r} + (\abs{W_{\Omega}} + n^2/D^4)\cdot \frac{r}{\abs{W_{\Omega}}} \right) \\
                                                                                                           & = \sum_\Omega \left((\abs{W_{\Omega}}\log{r} + r + \frac{n^2}{D^4}\cdot \frac{r}{\abs{W_{\Omega}}} \right) .
\end{align*}
Next, we recall that $r = \frac{c D^2}{n} \cdot \abs{W_{\Omega}}$ and $\sum_\Omega \abs{W_{\Omega}} =  O(n D\log{D} + n)$ (mentioned in the discussion preceding \eqref{eq:second_stage_cells}), in order to conclude that the above sum is $O_D(n)$, where the constant of proportionality depends polynomially on~$D$.

We conclude the construction by associating with each second-stage cell $\xi = \Omega \cap \sigma$ the set of curves
$\L_\xi \subset \L$ that it intersects, by testing for each curve $\gamma \in \L_\Omega$ whether it also meets the prism~$\sigma$. Omitting any further details, we have shown:

\begin{theorem}
  \label{thm:decomposition_alg}
  The decomposition described in Theorem~\ref{thm:cell_decomposition_restated} can be computed in randomized expected time~$O_D(n^2)$.
\end{theorem}

\begin{remark}
A major open problem is to improve the running time to subquadratic. The bottleneck is having to explicitly process $V(\L)$, or, more generally $V_k$, at each iteration $k$. In the worst case, this set could contain $\Theta(n^2)$ points. The remaining steps of the algorithm can be completed in~$O_D(n)$ time. Thus
the key to obtaining subquadratic running time lies in having an efficient implicit representation for $V(\L)$.
We next present such an efficient implementation, based on a range-search mechanism, for the case where $\L$ is a set of lines in $3$-space. 
\end{remark}

\subsection{A Faster Construction for the Case of Lines}
\label{sec:lines_alg}
In this section, we present an improved implementation of our algorithm for the case of lines in $\reals^3$ in general position, or, more generally, line segments in $3$-space.
To begin, we will present an algorithm that works for lines (or line segments) in non-vertical general position. We will then show how this can be extended to lines in general position.
Adapting our general position assumptions from the beginning of this section, this implies that no pair of projected lines (or line segments) coincide.

Our approach is to use a compact representation for the points of vertical visibility instead of storing them explicitly.
We note that, in contrast with the algorithm of Theorem~\ref{thm:cell_decomposition}, we 
are able to track only those pairs of vertically visible points that lie in the same cell of the current decomposition, we refer to them as \emph{unsplit visibility pairs}: pairs in which the two points end up in different cells are not tracked at all, we refer to them as \emph{split visibility pairs}.
We also comment that since the cells $\tau \in {\reals}^3 \setminus Z(f)$ may not be $xy$-monotone, we may also track unsplit visibility pairs where the two endpoints are not visible to each other inside $\tau$ (in case the two corresponding lines intersect in the $xy$-projection of $\tau$).
This, however, does not violate the analysis of Sections~\ref{sec:first_step} and~\ref{sec:second_step}, and the assertions in Theorem~\ref{thm:cell_decomposition} for the case of lines continue to hold.

We exploit the mechanism of Agarwal \cite{Agarwal-90} to efficiently represent (and count) intersections among line segments
in the plane.\footnote{The mechanism in~\cite{Agarwal-90} counts points with multiplicity in case there are three or more concurrent lines (or line segments). This, however, is not an issue in our analysis since the algorithm in~\cite{AMS-13} can handle points with multiplicity---see below.}
We revisit the algorithm of Agarwal, Matou{\v{s}}ek, and Sharir \cite{AMS-13} summarized in Theorem~\ref{thm:ams} and modify the procedures that were originally designed to manipulate the input points explicitly, to instead perform the required operations implicitly.
A closer inspection of the analysis in \cite{AMS-13} shows that we need to support the following two operations:
(i)~select uniformly a point at random among a collection of points contained in a specific cell
$\tau \in {\reals}^3 \setminus Z(f)$,
for an appropriate polynomial $f$, and
(ii)~count the number of points contained in~$\tau$.

Fix such a cell~$\tau$. As in Section~\ref{sec:second_step}, let $\L_{\tau}$ be the subset of lines meeting $\tau$. We take the intersection of each line in $\L_{\tau}$ with $\tau$ (this takes constant time in our model of computation),
obtain a collection~$S_\tau$ of open line segments contained in~$\tau$, and consider the set of their projections $S^{*}_{\tau}$ to the $xy$-plane.  Put $s_\tau \coloneqq |S^{*}_{\tau}|$. Using the algorithm in \cite{Agarwal-90} we construct a compact representation
for the pairwise intersecting segments in $S^{*}_{\tau}$ in overall $O(s_{\tau}^{4/3} \polylog{s_\tau})$ time. In particular, this
implies that operation (ii) can be completed in the same time bound. Concerning operation (i), the resulting compact representation
consists of a union of complete bipartite intersection graphs (each such graph is stored as a pair~$(A,B)$ of sets of segments, in which every segment of $A$ intersects every segment of $B$; the pair is stored using $\Theta(|A|+|B|)$ space rather than $\Theta(|A|\times |B|)$; hence the space savings).
Once such an implicit representation is available, it is possible to randomly sample a point of intersection in logarithmic time, by first picking the bipartite graph and then randomly and uniformly picking a segment from $A$ and a segment of $B$---see \cite{Agarwal-90} for more details concerning this construction.

The algorithm in~\cite{AMS-13} constructs a polynomial partitioning in several iterative steps. The majority of the algorithm's running time is spent on a procedure that computes a polynomial $f$ that simultaneously dissects\footnote{A polynomial $f$ \emph{dissects} a set $A$ if $f > 0$ on at most $(7/8)|A|$ points from $A$ and $f < 0$ on at most $(7/8)|A|$ points.} a collection $A_1, \ldots, A_k$ of sets of points, and also computes the sign of $f$ on each point of each set $A_i$. The polynomial $f$ is formed by lifting the points in $A_1, \ldots, A_k$ into a higher-dimensional Euclidean space using the \emph{Veronese mapping}, taking a small random sample of the points, and then computing the hyperplane passing through these randomly sampled points in the lifted space. The analysis in~\cite{AMS-13} shows that the composition of this hyperplane with the Veronese mapping is a polynomial that with high probability (with respect to the randomly chosen subsets of $A_1,\ldots,A_k$) dissects at least half of the sets $A_1, \ldots, A_k$. Recall that in the setting of our problem the points in $A_1, \ldots, A_k$ are not given explicitly. Instead we can use our efficient implementation of steps (i)--(ii) in order to construct $f$ as above, as well as counting how many points of vertical visibility among $\L$ (or, more generally, a collection of sets of line segments in $\reals^3$) are contained in the regions $\{f>0\}$ and $\{f<0\}$.
This is done over iterations as follows.
At iteration $j$ of the computation of the partitioning polynomial, we have $k \coloneqq 2^{j}$ sets of points of vertical visibility, each of which is represented as the disjoint union of complete bipartite graphs of line-segments (we begin with a trivial representation corresponding to the lines in $\L$).  We randomly sample points in the lifted space from each set  $A_1, \ldots, A_k$, compute the corresponding polynomial $f_j$, and then compute the sign of each point from each $A_i$, $i=1, \ldots, k$, w.r.t. $f_j$. In order to do so, we need to cut the line-segments participating in the representation of each $A_1, \ldots, A_k$  with $Z(f_j)$, and obtain a new collection of sets of line segments---these sets represent the points of vertical visibility at the next iteration $j+1$, that is, the new sets  $A'_1, \ldots, A'_{2k}$. We next compute a compact representation for each of these new sets in order to (i) count how many points of vertical visibility lie in the regions $\{f_j>0\}$ and $\{f_j<0\}$, and (ii) randomly sample points in $A'_1, \ldots, A'_{2k}$ in order to compute $f_{j+1}$.
While there are a few additional technical details, these issues do not impact the running time of the algorithm. Indeed, the overall expected running time is dominated by the total complexity of the compact representation for the pairwise intersecting segments, and is thus $O_D(n^{4/3}\polylog{n})$.

We once again emphasize that with this implementation we can only guarantee to control the number of unsplit visibility pairs
inside a cell (that is, both defining lines meet that cell and the two vertically visible points are contained in the cell).
The number of split visibility points within a cell
can be arbitrarily large. To summarize we have shown:

\begin{lemma}
  \label{lem:lines_alg}
  Let $\L$ be a collection of $n$ line segments in ${\reals}^3$ in non-vertical general position,
  and let $V(\L)$ be the set of points of their vertical visibilities.
  Let $D$ be a positive integer.  Then one can compute in expected $O_D(n^{4/3}\polylog{n})$ time
  a partitioning polynomial $f$ of degree $D$, such that each connected component of ${\reals}^3 \setminus Z(f)$
  contains $O(|V(\L)|/D^3)$ pairs of unsplit points of vertical visibility from~$V(\L)$.
\end{lemma}

We next describe the modifications to Lemma~\ref{lem:firstPartitioning} and Corollary~\ref{cor:firstPartitioning} needed to apply our algorithm.
Beginning with the first decomposition step, we observe that Lemma~\ref{lem:firstPartitioning} continues to hold if instead
of considering the entire set $V_k$, we consider only the subset of unsplit visibility pairs with respect to unacceptable cells, 
that is, those points of vertical visibility, for which both defining lines intersect the same unacceptable cell generated at step $k$.
With this refinement of $V_k$, we modify property~(C.2) in the statement of Lemma~\ref{lem:firstPartitioning} accordingly,
and inside $\Omega$ consider only the unsplit visibility pairs of $V_k$.
Then in the assertion of Corollary~\ref{cor:firstPartitioning} concerning $\Omega$ we can guarantee that either $\Omega$ intersects
$O(n/D^2)$ curves from $\L$, or $\Omega$ contains $O(n^2/D^4)$ pairs of (unsplit) points of vertical visibility from $V(\L)$, or both.

The implementation of the procedure to compute the partitioning polynomial $P$ (using the notation of
Corollary~\ref{cor:firstPartitioning}) is performed by repeatedly invoking Lemma~\ref{lem:lines_alg},
initially on the input lines in $\L$, and at step $k > 1$, on the set of the line segments obtained by intersecting the lines of~$\L$
with the unacceptable cells from step $k-1$ (this replaces the explicit representation of $V_{k-1}$).
At each step the number of line segments is only linear in $n$ and in $D$ (more specifically, every line is cut into at most $D+1$ segments),
and therefore the total running time of computing the partitioning polynomial for $V_k$, over all $O(\log D)$ iterations,
is $O_D(n^{4/3}\polylog{n})$.
In addition, we need to apply some of the operations already discussed above, including the classification of the cells as being
either acceptable or unacceptable;
this takes $O_D(n^{4/3}\polylog{n})$ time in total.

The execution of the second decomposition step proceeds verbatim as above, since we consider only the unsplit pairs of vertical
visibility in a cell $\Omega$.

Finally, we will remove the the assumption that the lines are non-vertical, by computing a partitioning polynomial $Q$ for the $xy$-projections
of the vertical lines using Theorem~\ref{thm:ams}, which takes $O_D(n)$ time in this case.
Then we take the product of $Q$ and $P$ and continue with the execution of the second decomposition step as just described.  
We thus conclude:

\begin{theorem}
  \label{thm:improved_alg_lines}
  The decomposition described in Theorem~\ref{thm:cell_decomposition_restated} for the case of $n$ lines or line segments in $\reals^3$ in general position can be computed in randomized expected time $O_D(n^{4/3}\polylog{n})$.
\end{theorem}

\section{An Application: Eliminating Depth Cycles among Lines}
\label{sec:applications}

In \cite{AS-18}, Aronov and Sharir obtained a combinatorial bound on the number of cuts needed to eliminate cycles in a collection of
pairwise-disjoint non-vertical lines in ${\reals}^3$.

The main obstruction to converting Aronov and Sharir's combinatorial bound into an algorithmic procedure was the absence of a constructive version of Theorem~\ref{thm:guth}.  More specifically, their proof proceeds by partitioning $\reals^3$ using a polynomial~$f$ of degree~$D$ (see more below on the choice of $D$), and then cutting each line not contained in $Z(f)$ at the points where it crosses $Z(f)$ (lines contained in $Z(f)$ need slightly different treatment; we omit the details here; this does not affect the asymptotics of the algorithm runtime or of the number of cuts required). This procedure produces at most $D$ cuts per line; these are cuts of the \emph{first type}. Every line $\ell$ is also cut at $O(D^2)$ additional points, which correspond to locations where $\ell$ passes above a critical point of $f$; more precisely,
this is a point $(x_0,y_0,z_0) \in \ell$ such that $f$ and $\partial f/\partial z$ are simultaneously zero at $(x_0,y_0,z_1)$ for some $z_1 < z_0$. This results in a total of $O(D^2n)$ \emph{second-type} cuts.

Now, for each connected component $\tau$ of $\reals^3 \setminus Z(f)$, Aronov and Sharir \cite{AS-18} collect the set $\L_\tau$ of the lines of $\L$ meeting $\tau$, and each set $\L_\tau$ is handled recursively, producing a recurrence of the form
\[
  C(n)=O(D^3) \cdot C(cn/D^2) + O(D^2 n)
\]
for the number of cuts $C(n)$ sufficient to eliminate all cycles in a set of $n$ lines, for a suitable absolute constant $c$.  The bound of $O(n^{3/2} \log^{O(1)}n)$ is obtained by setting $D$ to $\Theta(n^{1/4})$.

We next sketch how to efficiently implement the steps outlined above.  We construct the partitioning in time $O_D(n^{4/3}\polylog{n})$, using Theorem~\ref{thm:improved_alg_lines}, where we are now forced to choose $D$ a constant; our polynomial $P$ has degree at most $D\log D$, which increases the number of cuts to $O((D\log D)^2n)$ and the number of cells to $O((D\log D)^3)$.  Determining the first-type cuts of each line can be done in time $O_D(1)$ as described in Section~\ref{sec:alg_aspects}. 
Finding the cuts of the second type along a line can be done by constructing the solution set of the system $\{P = 0,\partial P/ \partial z =0\}$ in the vertical halfplane bounded by the line, in time $O_D(1)$; this follows from our assumption about the model of computation.
Additional work required to process the secondary subdivision involves simply cutting each line meeting a primary subdivision cell $\Omega$ (recall that the sets $\L_\Omega$ are constructed by the algorithm of Theorem~\ref{thm:improved_alg_lines}) at the points where it crosses the boundary of each prism $\sigma$, or equivalently finding the points where the projection of such a line enters and exits each trapezoid of the vertical decomposition $\VD(R^{*}_{\Omega})$; this computation is already performed in Theorem~\ref{thm:improved_alg_lines}.
We thus obtain $O(D^3\log^3{D})$ additional cuts for each line.

A close examination of \cite{AS-18} shows that, even though our partition is not a strictly polynomial one, due to the presence of the secondary subdivision, the correctness argument of Aronov and Sharir \cite{AS-18} carries through here as well.
Indeed, after the application of the first- and second-type cuts, as well as the cuts with the boundary of each prism $\sigma$, we are left to process the remaining cycles in each second-stage cell in recursion. In this case we collect all lines meeting a cell to form a new subproblem. 
To summarize, the number of cuts made by our algorithm is described by the recurrence
\[C(n)=O(D^3\log^3 D) \cdot C(cn/D^2) + O_D(n),\]
where $c$ is an absolute constant and $D$ is a constant of our choice.
The expected running time on the other hand is governed by the recurrence
\[
  T(n) =  O(D^3\log^3D) \cdot T(cn/D^2) + O_D(n^{4/3}\polylog{n}).
\]
They both solve to $O_D(n^{3/2+\eps(D)})$, once we pick a sufficiently large constant $D>0$; $\eps=\eps(D)>0$ depends on~$D$ and can be made arbitrarily small by increasing~$D$, so we can rewrite the bound as $O_{\eps}(n^{3/2 + \eps})$.
Note that, since $D$ cannot be set to grow with $n$, the number of cuts guaranteed by our algorithm is slightly larger than that guaranteed by the upper bound of \cite{AS-18}, namely $O(n^{3/2}\log^{O(1)}n)$.

\begin{theorem}
  \label{thm:line-cycles}
  Let $\L$ be a collection of $n$ pairwise-disjoint non-vertical lines in ${\reals}^3$ so that 
  no pair of lines have coinciding
  $xy$-projections. 
  Then one can apply $O_{\eps}(n^{3/2+\eps})$ cuts eliminating all depth cycles among the lines in $\L$.
  These cuts can be computed in expected time $O_{\eps}(n^{3/2+\eps})$, for any $\eps > 0$.
  The number of cuts is near optimal in the worst case.
\end{theorem}

\begin{remarks}
  (a) Previous algorithms that solve this problem apply an approximation algorithm of Aronov, De~Berg, Gray, and Mumford \cite{AdBGM-08}, which involves matrix multiplication. The running time is close to $O(n^{4 + 2\omega})$, where $\omega < 2.373$ is the exponent
  of matrix multiplication; this was later improved by De Berg \cite{deBerg-17} to $O(n^{3 + \omega})$.
  In spite of the fact our bound $O(n^{3/2+\eps})$ on the number of cuts is slightly inferior to the bound
  $O(n^{3/2}\log^{O(1)}n)$ in \cite{AS-18} as well as the bound resulting from \cite{AdBGM-08}, our algorithm is considerably more efficient.

  (b) Note that the algorithm described above works equally well with non-vertical pairwise-disjoint algebraic curves of constant degree, with only superficial modifications, mirroring the combinatorial analysis of Aronov and Sharir \cite{AS-18} as well as of Sharir and Zahl \cite{SZ-17}.  The current analysis, however, can only guarantee quadratic running time (see Theorem~\ref{thm:decomposition_alg}).
\end{remarks}

\paragraph*{Acknowledgments:}
    The authors wish to thank Micha Sharir for several useful discussions, and would especially like to recognize Saugata Basu for his patience and generosity answering the authors' questions on matters of algebra, and in referring them to the specific algorithms in \cite{BPR-06} cited in this paper. The authors would also like to thank the three anonymous referees for their numerous comments and suggestions.

\end{document}